\documentclass{article}

\usepackage[colorlinks,allcolors=blue]{hyperref}
\usepackage{amssymb,amsmath,amsthm,mathrsfs}
\usepackage{venturis2}
\usepackage[T1]{fontenc}
\def\B{\mathscr B}
\def\C{\mathbb C}
\def\d{\mathrm d}
\def\D{\mathscr D}
\def\F{\mathscr F}
\def\G{\mathcal G}
\def\H{\mathcal H}
\def\K{\mathscr K}
\def\M{\mathrm M}
\def\N{\mathbb N}
\def\R{\mathbb R}
\def\S{\mathbb S}
\def\SS{\mathscr S}

\def\Hrond{\mathscr H}
\def\hs{\mathfrak h}

\def\e{\mathop{\mathrm{e}}\nolimits}
\def\im{\mathop{\mathrm{Im}}\nolimits}
\def\re{\mathop{\mathrm{Re}}\nolimits}

\def\Ker{\mathop{\mathrm{Ker}}\nolimits}
\DeclareMathOperator*{\slim}{s\hspace{0.1pt}-\hspace{0.1pt}lim}
\def\ltwo{\mathsf{L}^{\:\!\!2}}
\def\Lp{\mathsf{L}^{\:\!\!p}}
\def\linf{\mathsf{L}^{\:\!\!\infty}}

\def\Ind{\mathop{\mathrm{Ind}}\nolimits}
\def\Tr{\mathop{\mathrm{Tr}}\nolimits}
\def\tr{\mathop{\mathrm{tr}}\nolimits}

\newtheorem{Theorem}{Theorem}[section]
\newtheorem{Remark}[Theorem]{Remark}
\newtheorem{Lemma}[Theorem]{Lemma}
\newtheorem{Corollary}[Theorem]{Corollary}
\newtheorem{Proposition}[Theorem]{Proposition}

\begin{document}


\title{Scattering operator and wave operators for 2D Schr\"odinger
operators with threshold obstructions}

\author{S. Richard$^1$\footnote{Supported by the grant \emph{Topological invariants
through scattering theory and noncommutative geometry} from Nagoya University,
and by JSPS Grant-in-Aid for scientific research  C no 18K03328 \& 21K03292, and on
leave of absence from Univ.~Lyon, Universit\'e Claude Bernard Lyon 1, CNRS UMR 5208,
Institut Camille Jordan, 43 blvd.~du 11 novembre 1918, F-69622 Villeurbanne cedex,
France.}, R. Tiedra de Aldecoa$^2$\footnote{Partially supported by the Chilean
Fondecyt Grant 1210003.}, L. Zhang$^3$}

\date{\small}
\maketitle
\vspace{-1cm}

\begin{quote}
\emph{
\begin{enumerate}
\item[$^1$] Graduate school of mathematics, Nagoya University,
Chikusa-ku,\\Nagoya 464-8602, Japan
\item[$^2$] Facultad de Matem\'aticas, Pontificia Universidad Cat\'olica de Chile,\\
Av. Vicu\~na Mackenna 4860, Santiago, Chile
\item[$^3$] School of Science, Nagoya University, Chikusa-ku, Nagoya 464-8602, Japan
\item[]E-mails: richard@math.nagoya-u.ac.jp, rtiedra@mat.uc.cl, zhly.ok@gmail.com
\end{enumerate}
}
\end{quote}


\begin{abstract}
We determine the low-energy behaviour of the scattering operator of two-dimensional
Schr\"odinger operators with any type of obstructions at $0$-energy. We also derive
explicit formulas for the wave operators in the absence of p-resonances, and outline
in this case a topological version of Levinson's theorem.
\end{abstract}


\smallskip



\section{Introduction and main results}\label{section_intro}
\setcounter{equation}{0}

\subsection{The analytic study}

Scattering theory for two-dimensional Schr\"odinger operators is a challenging subject
that has been the focus of many studies. And since the list of papers dealing with it
is very long, we will mention here only the ones relevant for our work. The story started
with a {\em surprise} put into evidence in \cite{BGDW}: scattering properties of
two-dimensional Schr\"odinger operators are very different from their
three-dimensional analogues. The main difference is due to resonances at $0$-energy,
which can be divided into two categories: s-resonances (at most one) and p-resonances
(at most two). It was shown in \cite{BGDW}, and later confirmed in
\cite{BGD88}, that the presence of a s-resonance does not play an important role,
while the presence of p-resonances has surprising consequences. For instance, they
lead to a contribution of value $1$ in the so-called Levinson's theorem, in a way
similar to usual bound states. Unfortunately, the proofs of the results of
\cite{BGD88} are based on double asymptotic expansions of the resolvent, which make
them strenuous to follow.

A decade later, a renewed interest in the two-dimensional case has been triggered by
the works \cite{JY02,Yaj99} on the $\Lp$-boundedness of the wave operators. However,
these works were conducted under the assumption that $0$-energy bound states and
$0$-energy resonances are absent (the so-called regular or generic case). The next
breakthrough came with the derivation in \cite{JN01} of a simplified resolvent
expansion, no longer given as a two parameters expansion, but in terms of powers of a
single parameter. Subsequently, numerous works took advantage of this simplified
resolvent expansion, as for example \cite{B16,EG12_0,RT13_2,Sch05} in which the
assumption of absence $0$-energy bound states and $0$-energy resonances remains. In
other works, it was assumed that $0$-energy bound states and p-resonances are absent,
as for example in \cite{T17}, or only that the p-resonances are absent, as in
\cite{EGG}. Note however that the behaviour of the Schr\"odinger evolution group has
been studied in the general case in \cite{EG12_1}. More recently, two-dimensional
Schr\"odinger operators with point interactions have been investigated: the boundedness of
the wave operators in $\Lp$-spaces in the regular case has been discussed in
\cite{CMY}, while a full picture has been provided in \cite{Yaj20_1}. Building on the
latter, $\Lp$-boundedness for more general Schr\"odinger operators with threshold
obstructions has been studied in \cite{Yaj20_2}.

The present paper is a continuation of our work \cite{RT13_2} on the wave operators
for two-dimensional Schr\"odinger operators. In that paper, we considered the regular
case. Here, we do not make this assumption anymore, and present several results either
in the general case, or under the assumption of absence of p-resonances only. Our
results are in line with the ones obtained in \cite{BGDW}, but we present them in an
updated (and presumably simpler) framework. 

We consider the scattering system given by a pair of operators $(H,H_0)$, where $H_0$
is the Laplacian in the Hilbert space $\ltwo(\R^2)$ and $H:=H_0+V$ with $V$ a real
potential decaying rapidly at infinity. Under quite general conditions on $V$ it is
known that the wave operators
$$
W_\pm:=\slim_{t\to\pm\infty}\e^{itH}\e^{-itH_0}
$$
exist and are complete. As a consequence, the scattering operator $S:=W_+^*W_-$ is
unitary in $\ltwo(\R^2)$. Since $S$ strongly commutes with $H_0$, the operator $S$
decomposes in the spectral representation of $H_0$, meaning that $S$ is unitarily
equivalent to a family of unitary operators $\{S(\lambda)\}_{\lambda\in\R_+}$ in
$\ltwo(\S)$. For historical reasons, the operator $S(\lambda)$ is called the
scattering matrix at energy $\lambda$, even though it acts on an infinite-dimensional
Hilbert space. A function $f:\R^2\to\C$ satisfying $Hf=0$ in the distributional sense
is called a $0$-energy bound state if $f$ belongs to the domain of $H$. If
$f\in\linf(\R^2)$, then $f$ is called a s-resonance, while if $f\in\Lp(\R^2)$ for some
$p\in(2,\infty)$, then $f$ is called a p-resonance. These three distinct cases are
related to orthogonal projections $S_1\ge S_2\ge S_3$ in $\ltwo(\R^2)$ introduced in
\cite{JN01}: $S_3\ne0$ if there are $0$-energy bound states, $T_3:=S_2-S_3\ne0$ if
there are p-resonances, and $T_2:=S_1-S_2\ne0$ if there is a s-resonance. In this
setup, the generic case corresponds to the assumption $S_1=0$.

Our first main result concerns the low-energy behaviour of the scattering matrix. It
holds in the general case, without any assumption on the absence of $0$-energy bound
states or resonances. A more detailed version of the result is presented in Theorem
\ref{thm_S_matrix}.

\begin{Theorem}[Scattering matrix at $0$-energy]\label{thm11}
If $V$ satisfies $|V(x)|\le{\rm Const.}\;\!\langle x\rangle^{-\rho}$ for a.e.
$x\in\R^2$ and $\rho>11$, then $\lim_{\lambda\searrow0}S(\lambda)=1$.
\end{Theorem}

Let us stress that the behaviour $\lim_{\lambda\searrow0}S(\lambda)=1$ is completely
different from its counterparts in one and three dimensions, where the value of the
scattering matrix $S(0)$ depends on the presence or absence of resonances at
$0$-energy (see \cite{AK,JK79}). We also acknowledge that a similar result has been
established before in \cite{BGD88}, case by case for the different types of
obstructions at $0$. However, it was done under the much stronger assumption of
exponential decay of the potential and the condition $\int_{\R^2}V(x)\;\!\d x\ne0$.
Clearly, the extension of this result to a power law decaying potential was expected
by the community of experts, but up to the best of our knowledge, it has never been
explicitly proved.

\begin{Remark}
For the proof of Theorem \ref{thm11}, we rely on the resolvent expansion provided in
the seminal paper \cite{JN01}, and the decay assumption on $V$ is prescribed by this
approach. By using more recent tools, as for example the ones developed in
\cite{Yaj20_2}, it is certainly possible to weaken the condition $\rho>11$. However,
our aim in this paper was not only to provide the low energy behavior of the
scattering matrix, but also to get explicit formulas for the wave operators. For that
purpose, relying on a classical approach was an asset. As a comparison, in the
one-dimensional case it took about 12 years between the derivation of the first closed
formulas for the wave operators and their proof for potentials with optimal decay, see
\cite {I,KR08}. 
\end{Remark}
 
Our second main result is an explicit formula for the wave operator $W_-$ (a similar
formula for $W_+$ can be obtained by using the relation $W_+=W_-S^*$). The formula is
obtained in Theorem \ref{thm_wave_op}, but we present here the version to be found in
Corollary \ref{cor_after_com}. For its statement, we use the notation $A$ for the
generator of dilations in $\ltwo(\R^2)$ and $\K(\ltwo(\S))$ for the set of compact
operators on $\ltwo(\S)$. Also, given a continuous function
$\eta:\R_+\to\K(\ltwo(\S))$ and the operator
$\F_0:\ltwo(\R^2)\to\ltwo(\R_+;\ltwo(\S))$ which diagonalises $H_0$ (see
\eqref{eq_diag}), we write $\eta(H_0)$ for the bounded operator in $\ltwo(\R^2)$
satisfying
$$
\big(\F_0\;\!\eta(H_0)f\big)(\lambda)
=\eta(\lambda)(\F_0 f)(\lambda),
\quad\hbox{$f\in\ltwo(\R^2)$, a.e. $\lambda\in\R_+$.}
$$

\begin{Theorem}\label{Thm_cor}
Let $V$ satisfy $|V(x)|\le{\rm Const.}\;\!\langle x\rangle^{-\rho}$ for a.e.
$x\in\R^2$ and $\rho>11$, and $T_3=0$. Then, there exist two continuous functions
$\eta,\widetilde\eta:\R_+\to\K(\ltwo(\S))$ vanishing at $0$ and $\infty$, satisfying
$\eta(H_0)+\widetilde\eta(H_0)=S-1$, and such that
\begin{equation}\label{eq_start}
W_--1=\tfrac12\big(1+\tanh(\pi A/2)\big)\eta(H_0)
+\tfrac12\big(1+\tanh(\pi A)-i\cosh(\pi A)^{-1}\big)\widetilde\eta(H_0)+K
\end{equation}
with $K\in\K(\ltwo(\R^2))$.
\end{Theorem}

Note that we have not been able to obtain a similar formula in the general case with
no assumption on the absence of p-resonances. Let us also mention that if there is no
$0$-energy bound state (i.e. $S_3=0$), then $\widetilde\eta(H_0)=0$ and
$\eta(H_0)=S-1$.

\subsection{The topological outcomes}

Theorems \ref{thm11} and \ref{Thm_cor} contain the main results of this paper, and
their proofs are purely analytic. However, the motivation for getting these results
comes from a topological index theorem which was out of reach at the time of
\cite{BGDW}. In the second part of this introduction, we briefly outline some
corollaries which follow from Theorem \ref{Thm_cor}, and then compare these results
with related results  in the literature. Additional information and results about
index theorems in scattering theory can be found in the review paper \cite{Ric16}. 

Formula \eqref{eq_start} shows that the wave operator $W_-$ coincides, up to a compact
operator, with a combination of functions of $A$ and $H_0$. This is even more apparent
if one looks at the expression for $\F_0(W_--1)\F_0^*$ obtained in Corollary
\ref{cor_vanish}. In that representation, the operator $W_--1$ can be expressed as a
combination of functions of the generator $A_+$ of dilations on $\R_+$ and functions
of the position operator on $\R_+$ with values in $\K(\ltwo(\S))$ (multiplication
operators on $\R_+$ with values in $\K(\ltwo(\S))$). And these functions are not
arbitrary: the functions of $A_+$ have limits at $\pm\infty$ and the functions of the
position operator vanish at $0$ and $\infty$. A $C^*$-algebra generated by such
functions has been extensively studied in \cite[Sec.~4.4]{Ric16}; the multiplication
operators were taking values in $\C$ instead of $\K(\ltwo(\S))$, but a tensor product
with the ideal $\K(\ltwo(\S))$ leads to the functions appearing here. In particular
the $K$-theory of this algebra has been determined, and a topological version of
Levinson's theorem was illustrated with several examples. Therefore, our goal here is
not to recall the $C^*$-algebraic machinery, but to present its consequences for our
model.

The operator $W_-$ is a Fredholm operator with trivial kernel and cokernel spanned by
the eigenfunctions of $H$. Its Fedholm index is equal to minus the number of
eigenvalues of $H$, multiplicity counted. Algebraically, this means that the
$K_0$-element associated to $W_-$ corresponds to (minus) the projection $E_{\rm p}(H)$
on the subspace spanned by the eigenfunctions of $H$. Being a Fredholm operator, the
image of $W_-$ in the Calkin algebra (the quotient of $\B(\ltwo(\R^2))$ by the ideal
of compact operators $\K(\ltwo(\R^2))$) is a unitary operator. Of course, this algebra
is too complicated to do anything with it, but the interest of formula
\eqref{eq_start} is precisely that it allows to compute explicitly the image. Indeed,
by evaluating the functions $\tfrac12(1+\tanh(\cdot))$ and
$\tfrac12(1+\tanh(\cdot)-i\cosh(\cdot)^{-1})$ at $+\infty$, and by using that
$\eta(H_0)+\widetilde\eta(H_0)=S-1$ with $S(0)=S(\infty)=1$, one can identify the
image of $W_-$ with the family $\{S(\lambda)\}_{\lambda\in\R_+}$ (see
\cite[Sec.~4.4]{Ric16}). Collecting what precedes, one ends up with the $K$-theoretic
equation
\begin{equation}\label{eq_topol}
\Ind\big[\{S(\lambda)\}_{\lambda\in\R_+}\big]_1=-[E_{\rm p}(H)]_0.
\end{equation}
This equality corresponds to a topological version of Levinson's theorem, an example
of index theorem in scattering theory. However, it is not an equality between numbers
yet. It is a relation between an equivalence class of unitary operators and an
equivalence class of projections, with $\Ind$ the index map of $K$-theory. One could
stop the presentation here, because a relation between equivalence classes of objects
is stronger than an equality between numbers. However, Levinson's theorem is usually
presented as an equality between numbers, so let us go one step further.

Extracting numbers from an equation like \eqref{eq_topol} requires the use of
$n$-traces from cyclic cohomology. Fortunately, this is rather direct in our
situation. For the right-hand side, the only way to get a number is to consider the
usual trace $\Tr$ on $\K(\ltwo(\R^2))$. Applying it gives (minus) the number of
eigenvalues of $H$, multiplicity counted. For the left-hand side, the corresponding
$1$-trace is nothing but the winding number. However, since $S(\cdot)$ takes values
not in $\C$ but in $\C+\K(\ltwo(\S))$, this winding number has to be regularised. 
Explanations about the regularisation procedure are given in the Appendix of
\cite{KR12}. In our case, it amounts to check that an appropriate analytic formula can
be applied to the representative $\{S(\lambda)\}_{\lambda\in\R_+}$. Doing so, one gets
the numerical equality
\begin{equation}\label{eq_nous}
\tfrac1{2\pi}\int_0^\infty
\tr\big(i(1-S(\lambda))^nS(\lambda)^*S'(\lambda)\big)\;\!\d\lambda
=\Tr\big(E_{\rm p}(H)\big),
\end{equation}
where $\tr$ is the usual trace on $\K(\ltwo(\S))$, $n$ a sufficiently large integer,
and $S'(\cdot)$ the derivative of $S(\cdot)$. We do not check here that $S(\cdot)$ is
differentiable under our assumption on $V$, nor look for the smallest integer $n$ for
which \eqref{eq_nous} holds. But this can be done with the tools developed in the
following sections.

Another version of Levinson's theorem was obtained in \cite[Thm.~6.3]{BGD88} under the
assumption of exponential decay of the potential and the condition
$\int_{\R^2}V(x)\;\!\d x\ne0$, but without assumption on the absence of p-resonances.
In the framework of \cite{BGD88}, Levinson's theorem is expressed as
\begin{equation}\label{eq_eux}
\int_0^\infty\im\big((H-\lambda-i0)^{-1}-(H_0-\lambda-i0)^{-1}\big)\;\!\d\lambda
=-N_-+\pi \;\!\Delta_{-1,-1}-\tfrac14\int_{\R^2}V(x)\;\!\d x,
\end{equation}
where $N_-$ is the number of strictly negative eigenvalues of $H$ and
$\Delta_{-1,-1}$ an integer related to the $0$-energy eigenvalues and p-resonances.
Clearly, the relations \eqref{eq_nous} and \eqref{eq_eux} cannot be directly compared,
even if one takes into account the formal identity \cite[Eq.~(6.45)]{BGD88}:
$$
\im\Tr\big((H-\lambda-i0)^{-1}-(H_0-\lambda-i0)^{-1}\big)
=-\tfrac i2\tfrac\d{\d\lambda}\Tr\big(\ln(S(\lambda))\big).
$$
However, the common and important feature of these two relations is that the presence
of a s-resonance does not contribute to Levinson's theorem. This point is surprising
when compared with the one and three-dimensional cases. On the other hand, the
relation \eqref{eq_eux} implies that each p-resonance leads to a contribution of value
$1$ to Levinson's theorem, like a $0$-energy eigenvalue. A similar information cannot
be inferred from \eqref{eq_nous} since p-resonances have been excluded in our
analysis.

In the following final remark, we make some comparisons between our results and the
content of the recent papers \cite{CMY,Yaj20_1,Yaj20_2}.

\begin{Remark}\label{rem_comparison}
For two-dimensional Schr\"odinger operators with potential given by $N$ point
interactions, it is shown in \cite[Lemmas 3.1 \& 3.2]{Yaj20_1} that
\begin{equation}\label{eq_CMY}
W_--1
=\sum_{j,k=1}^N\tau_{y_j}K\big(\Gamma(\sqrt{\Delta})^{-1}\big)_{jk}\;\!\tau^*_{y_k},
\end{equation}
with $\tau_y$ the shift operator by $y\in\R^2$, $y_j\in\R^2$ the positions of the
point interactions, $(\Gamma(\sqrt{\Delta})^{-1})_{jk}$ the entries of the inverse of
a matrix-valued Fourier multiplier, and $K$ a singular integral operator in
$\ltwo(\R^2)$ with kernel
$$
K(x,y)=\tfrac2{\pi^2i}\tfrac1{x^2-y^2+i0}\;\!,\quad x,y\in\R^2.
$$
A similar formula for $W_+$ is also given in \cite[Sec.~4]{CMY}. If one uses polar
coordinates and writes $P_0$ for the orthogonal projection in $\ltwo(\S)$ on the
constant functions, one can show as in \cite[Thm.~2.5]{RT13_2} that the operator $K$
satisfies
\begin{equation}\label{eq_K}
K=-4\;\!\tfrac12\big(1+\tanh(\pi A/2)\big)(1\otimes P_0).
\end{equation}
Then, by inserting \eqref{eq_K} into \eqref{eq_CMY}, by evaluating the function
$\tfrac12(1+\tanh(\cdot))$ at $+\infty$, and by comparing the resulting operator with
the expression for the scattering operator given in \cite[Eq.~(II.4.35)]{AGHH}, one
obtains that $W_--1$ is essentially a combination of the function
$\tfrac12(1+\tanh(\pi A/2))$ and the scattering operator. This result, which was
first obtained for $1$ point interaction in \cite{KR}, is similar to
the content of Theorem \ref{Thm_cor}.

As emphasized in \cite{CMY} and \cite{Yaj20_1}, this picture is accurate in the regular
case. But \eqref{eq_CMY} ceases to be valid when threshold singularities are present:
in such a case, the map $\lambda\mapsto(\Gamma(\lambda)^{-1})_{jk}$ exhibits
singularities as $\lambda\searrow0$, and the boundedness of each individual term is no
more satisfied (the same issue leads to the introduction of the second
term in \eqref{eq_start} when $0$-energy bound states are present). In
\cite[Lemma 4.4]{Yaj20_1}, it is shown how the summands in \eqref{eq_CMY} have to be
combined in order to obtain bounded operators. And the same procedure is applied to
more general potentials in \cite[Lemma 5.24]{Yaj20_2}. However, the computations
involve PDE technics, and so far it has not been possible to extract any closed
formula from them. A closed formula in the presence of $p$-resonances would certainly
contain one more term on the r.h.s. of \eqref{eq_start}, and lead to a new
contribution in a topological version of Levinson's theorem, as displayed in
\eqref{eq_eux}.
We refer to \cite{RT20,RU} for other examples of singular integral operators that
have been shown to be equal to nice functions of simpler operators.
\end{Remark}

\noindent
{\bf Notations:} $\N:=\{0,1,2,\ldots\}$ is the set of natural numbers, $\SS$ the
Schwartz space on $\R^2$, $\R_+:=(0,\infty)$, and
$\langle\cdot\rangle:=\sqrt{1+|\cdot|^2}$. The sets $\H^s_t$ are the weighted
Sobolev spaces over $\R^2$ with index $s\in\R$ for derivatives and index $t\in\R$ for
decay at infinity \cite[Sec.~4.1]{ABG}, and with shorthand notations $\H^s:=\H^s_0$,
$\H_t:=\H^0_t$, and $\H:=\H^0_0=\ltwo(\R^2)$. For any $s,t\in\R$, the $2$-dimensional
Fourier transform $\F$ is a topological isomorphism of $\H^s_t$ onto $\H^t_s$, and the
scalar product $\langle\cdot,\cdot\rangle_\H$ (antilinear in the first argument)
extends continuously to a duality $\langle\cdot,\cdot\rangle_{\H^s_t,\H^{-s}_{-t}}$
between $\H^s_t$ and $\H^{-s}_{-t}$. Given two Banach spaces $\G_1$ and $\G_2$,
$\B(\G_1,\G_2)$ (resp. $\K(\G_1,\G_2)$) denotes the set of bounded (resp. compact)
operators from $\G_1$ to $\G_2$, with shorthand notation $\B(\G_1):=\B(\G_1,\G_1)$
(resp. $\K(\G_1):=\K(\G_1,\G_1)$). Finally, $\otimes$ stands for the closed tensor
product of Hilbert spaces or of operators.

\section{Preliminaries}\label{sec_prelim}
\setcounter{equation}{0}

\subsection{Free Hamiltonian}\label{sec_free_ham}

Set $\hs:=\ltwo(\S)$ and $\Hrond:=\ltwo(\R_+;\hs)$, and let $H_0$ be the (positive)
self-adjoint operator in $\H=\ltwo(\R^2)$ given by minus the Laplacian $-\Delta$ on
$\R^2$. Then, the unitary operator $\F_0:\H\to\Hrond$ defined by
\begin{equation}\label{eq_diag}
\big((\F_0 f)(\lambda)\big)(\omega)=2^{-1/2}(\F f)(\sqrt\lambda\;\!\omega),
\quad f\in\SS,~\lambda\in\R_+,~\omega\in\S,
\end{equation}
is a spectral transformation for $H_0$ in the sense that
$$
(\F_0H_0f)(\lambda)
=\lambda\;\!(\F_0 f)(\lambda)
=(L\F_0 f)(\lambda),
\quad\hbox{$f\in\H^2$, a.e. $\lambda\in\R_+$,}
$$
with $L$ the maximal multiplication operator by the variable $\lambda\in\R_+$ in
$\Hrond$. Moreover, for each $\lambda\in\R_+$, the operator $\F_0(\lambda):\SS\to\hs$
given by $\F_0(\lambda)f:=(\F_0f)(\lambda)$ extends to an element of $\B(\H^s_t,\hs)$
for any $s\in\R$ and $t>1/2$, and the function
$\R_+\ni\lambda\mapsto\F_0(\lambda)\in\B(\H^s_t,\hs)$ is continuous (see Lemma
\ref{lemma_F_0} for additional continuity properties of $\F_0(\cdot)$).

The asymptotic expansion of $\F_0(\lambda)$ as $\lambda\searrow0$ will play an
important role in our $2$-dimensional case, in a similar way it does in the
$3$-dimensional case \cite[Sec.~5]{JK79}. By expanding the exponential
$\e^{-i\sqrt\lambda\omega\cdot x}$ in Taylor series, one gets
\begin{equation}\label{eq_exp_F_0}
\F_0(\lambda)=\gamma_0+\sqrt\lambda\;\!\gamma_1+\lambda\gamma_2+o(\lambda),
\quad\lambda\in\R_+,
\end{equation}
with $\gamma_j:\SS\to\hs$ ($j=0,1,2$) the operator given by
$$
(\gamma_jf)(\omega)
:=\tfrac{(-i)^j}{2^{3/2}\pi\;\!(j!)}\int_{\R^2}\d x\,(\omega\cdot x)^j\;\!f(x),
\quad f\in\SS,~\omega\in\S.
$$
One can check that $\gamma_j$ extends to an element of $\B(\H^s_t,\hs)$ for any
$s\in\R$ and $t>j+1$, which implies that the expansion \eqref{eq_exp_F_0} holds in
$\B(\H^s_t,\hs)$ as $\lambda\searrow0$ for any $s\in\R$ and $t>3$. We shall sometimes
use the abbreviated notation $\gamma_2(\lambda)$, or $O(\lambda)$, for the sum
$\lambda\gamma_2+o(\lambda)$ in \eqref{eq_exp_F_0}.

\subsection{Perturbed Hamiltonian}\label{sec_perturbed}

Let us now consider a potential $V\in\linf(\R^2;\R)$ satisfying for some $\rho>1$ the
bound
\begin{equation}\label{eq_cond_V}
|V(x)|\le{\rm Const.}\;\!\langle x\rangle^{-\rho},\quad\hbox{a.e. $x\in\R^2$.}
\end{equation}
Then, the perturbed Hamiltonian $H:=H_0+V$ is a short range perturbation of $H_0$, and
it is known that the corresponding wave operators
$$
W_\pm:=\slim_{t\to\pm\infty}\e^{itH}\e^{-itH_0}
$$
exist and are complete. As a consequence, the scattering operator $S:=W_+^*W_-$ is
unitary in $\H$. Now, define for $z\in\C\setminus\R$ the resolvents of $H_0$ and $H$
$$
R_0(z):=(H_0-z)^{-1}\quad\hbox{and}\quad R(z):=(H-z)^{-1}.
$$
In order to recall properties of $R_0(z)$ and $R(z)$ as $z$ approaches the real axis,
it is convenient to decompose the potential $V$ according to the following rule: for
a.e. $x\in\R^2$ set
$$
v(x):=|V(x)|^{1/2}
\quad\hbox{and}\quad
u(x):=
\begin{cases}
+1 & \hbox{if $V(x)\ge0$}\\
-1 & \hbox{if $V(x)<0$,}
\end{cases}
$$
so that $u$ is self-adjoint and unitary and $V=uv^2$. Then, using the fact that $H$
has no positive eigenvalues \cite[Sec.~1]{Kat59} and that a limiting absorption
principle holds for $H_0$ and $H$ \cite[Thm.~4.2]{Agm75}, we infer that the limits
$$
vR_0(\lambda\pm i0)v:=\lim_{\varepsilon\searrow0}vR_0(\lambda\pm i\varepsilon)v
\quad\hbox{and}\quad
vR(\lambda\pm i0)v:=\lim_{\varepsilon\searrow0}vR(\lambda\pm i\varepsilon)v,
$$
exist in $\B(\H)$ and are continuous in the variable $\lambda\in\R_+$. This, together
with the relation
$$
u-uvR(\lambda\pm i\varepsilon)vu=\big(u+vR_0(\lambda\pm i\varepsilon)v\big)^{-1},
\quad\lambda\in\R_+,~\varepsilon>0,
$$
implies the existence and the continuity of the function
$\R_+\ni\lambda\mapsto(u+vR_0(\lambda\pm i0)v)^{-1}\in\B(\H)$. Furthermore, one has
$\lim_{\lambda\to\infty}(u+vR_0(\lambda\pm i0)v)^{-1}=u$ in $\B(\H)$, since
$\lim_{\lambda\to\infty}vR_0(\lambda+i0)v=0$ in $\B(\H)$ \cite[Prop.~7.1.2]{Yaf10}. On
the other hand, the existence in $\B(\H)$ of the limits
$\lim_{\lambda\searrow0}(u+vR_0(\lambda\pm i0)v)^{-1}$ depends on the presence or
absence of eigenvalues or resonances at $0$-energy. This problem has been studied in
detail in \cite{JN01} in dimensions $1$ and $2$. We recall here the main result in
dimension $2$ \cite[Thm.~6.2(ii)]{JN01}: Take $\kappa\in\C^*$ with $\re(\kappa)\ge0$,
let $\eta:=1/\ln(\kappa)$ (with $\ln$ the principal value of the complex logarithm),
and set
$$
\M(\kappa):=u+vR_0(-\kappa^2)v.
$$
Then, if $V$ satisfies \eqref{eq_cond_V} with $\rho>11$ and if $0<|\kappa|<\kappa_0$
with $\kappa_0>0$ small enough, the operator $\M(\kappa)^{-1}$ admits an expansion
\begin{equation}\label{eq_JN}
\M(\kappa)^{-1}
=I_1(\kappa)-g(\kappa)I_2(\kappa)-\tfrac{g(\kappa)\eta}{\kappa^2}\;\!I_3(\kappa),
\end{equation}
with
\begin{align}
I_1(\kappa)&:=(\M(\kappa)+S_1)^{-1},\label{eq_I_1}\\
I_2(\kappa)&:=(\M(\kappa)+S_1)^{-1}S_1(M_1(\kappa)+S_2)^{-1}S_1
(\M(\kappa)+S_1)^{-1},\label{eq_I_2}\\
I_3(\kappa)&:=(\M(\kappa)+S_1)^{-1}S_1(M_1(\kappa)+S_2)^{-1}S_2
\big(T_3m(\kappa)^{-1}T_3-T_3 m(\kappa)^{-1}b(\kappa)d(\kappa)^{-1}S_3\nonumber\\
&\quad-S_3d(\kappa)^{-1}c(\kappa)m(\kappa)^{-1}T_3+S_3d(\kappa)^{-1}c(\kappa)
m(\kappa)^{-1}b(\kappa)d(\kappa)^{-1}S_3+S_3d(\kappa)^{-1}S_3\big)\nonumber\\
&\quad\cdot S_2(M_1(\kappa)+S_2)^{-1}S_1(\M(\kappa)+S_1)^{-1},
\label{eq_I_3}
\end{align}
and where $S_1\ge S_2\ge S_3$ are orthogonal projections in $\H$, $T_3:=S_2-S_3$,
$g:\C\to\C$ satisfies $g(\kappa)=O(\eta^{-1})$ for $0<|\kappa|<\kappa_0$,
$m:\C\to\B(\H)$ satisfies $m(\kappa)=O(\eta^{-1})$ for $0<|\kappa|<\kappa_0$, and all
other factors are operator-valued functions having limits in $\B(\H)$ as $\kappa\to0$.
Precise formulas for these factors are provided in \cite[Sec.~6]{JN01} and will be
recalled in due time.

\section{Scattering operator}\label{sec_scattering}
\setcounter{equation}{0}

In this section we analyse the behaviour at low energy of the scattering matrix. Since
the scattering operator $S$ strongly commutes with $H_0$, it decomposes in the
spectral representation of $H_0$. That is, there exist for a.e. $\lambda\in\R_+$ a
unitary operator $S(\lambda)\in\B(\hs)$ such that
$$
(\F_0 S\F_0^*\varphi)(\lambda)=S(\lambda)\varphi(\lambda),
\quad\hbox{$\varphi\in\Hrond$, a.e. $\lambda\in\R_+$.}
$$
Furthermore, if $V$ satisfies \eqref{eq_cond_V} for some $\rho>1$, then the operators
$S(\lambda)$ are given by the stationary formula \cite[Thm.~1.8.1]{Yaf10}
\begin{equation}\label{eq_S_matrix}
S(\lambda)
=1_\hs-2\pi i\F_0(\lambda)v\big(u+vR_0(\lambda+i0)v\big)^{-1}v\F_0(\lambda)^*,
\quad\hbox{a.e. $\lambda\in\R_+$.}
\end{equation}
Our goal consists in determining the behaviour of this expression as
$\lambda\searrow0$. To this end, we will use both the expansion \eqref{eq_exp_F_0} for
$\F_0(\lambda)$ and the asymptotic expansion \eqref{eq_JN} for $\M(\kappa)^{-1}$ in
the case $\kappa=-i\sqrt\lambda$. This choice of $\kappa$ corresponds to the value
$\lambda+i0=-\kappa^2$ appearing in \eqref{eq_S_matrix}. For the sake of brevity, we
will keep using the shorthand notations
\begin{equation}\label{eq_convention}
\kappa:=-i\sqrt\lambda
\quad\hbox{and}\quad
\eta:=\tfrac1{\ln(\kappa)}=\tfrac1{\ln(\lambda)/2-i\pi/2}.
\end{equation}
From now on, we thus assume that $V$ satisfies \eqref{eq_cond_V} with $\rho>11$, so
that both expansions are verified (the expansion for $\F_0(\lambda)$ is verified
because the operator $v$ in $\F_0(\lambda)v$ satisfies $v\in\B(\H,\H_t)$ with $t>3$).
As a consequence, the problem reduces to computating the limit $\lambda\searrow0$ of
the operator
\begin{equation}\label{eq_to_prove}
\big(\gamma_0+\sqrt\lambda\;\!\gamma_1+\gamma_2(\lambda)\big)
v\M(\kappa)^{-1}v
\big(\gamma_0^*+\sqrt\lambda\;\!\gamma_1^*+\gamma_2(\lambda)^*\big)
\end{equation}
with $\M(\kappa)^{-1}$ given by \eqref{eq_JN}. Since the computation requires various
preparatory lemmas, we start by stating the final result, and then proceed to its
proof:

\begin{Theorem}[Scattering matrix at $0$-energy]\label{thm_S_matrix}
If $V$ satisfies \eqref{eq_cond_V} with $\rho>11$, then
$\lim_{\lambda\searrow0}S(\lambda)=1_\hs$ in $\B(\hs)$.
\end{Theorem}

For our first lemma, we need to recall the definition of two orthogonal projections
introduced in \cite[Sec.~6]{JN01}:
$$
P:=\tfrac1{\|v\|^2_\H}|v\rangle\langle v|\quad\hbox{and}\quad Q:=1-P,
$$
where $|v\rangle\langle v|f:=\langle v,f\rangle_\H\;\!v$ for any $f\in\H$. We also
need the vector notation $X=(X_1,X_2)$, with $X_j$ the maximal multiplication operator
in $\H$ by the $j$-th variable in $\R^2$.

\begin{Lemma}\label{lemma_Q}~
\begin{enumerate}
\item[(a)] One has $\gamma_0vQ=0=Qv\gamma_0^*$.
\item[(b)] For $j=1,2,3$, one has $\gamma_0vS_j=0=S_jv\gamma_0^*$.
\item[(c)] One has $\gamma_1vS_3=0=S_3v\gamma_1^*$.
\item[(d)] For $j=1,2,3$, one has $PS_j=0=S_jP$.
\end{enumerate}
\end{Lemma}

\begin{proof}
(a) For any $f\in\H$, we have
$$
2^{3/2}\pi\gamma_0vQf
=\int_{\R^2}\d x\,v(x)(Qf)(x)
=\int_{\R^2}\d x\,v(x)\Big(f(x)-\tfrac{v(x)}{\|v\|^2_\H}\langle v,f\rangle_\H\Big)=0,
$$
which proves the first equality. The second equality is obtained by duality.

(b) The claim follows from point (a) and the fact that $Q\ge S_j$ for $j=1,2,3$ (see
\cite[Thm.~6.2(i)]{JN01}).

(c) For any $f\in\H$ and $\omega\in\S^1$, we have
$$
2^{3/2}\pi i(\gamma_1vS_3f)(\omega)
=\int_{\R^2}\d x\,(\omega\cdot x)v(x)(S_3f)(x)
=\sum_{j=1}^2\omega_j\;\!\big\langle v,X_jS_3f\big\rangle_\H
=0,
$$
where the last equality follows from \cite[Eq.~(6.100)]{JN01}. This proves the first
equality. The second equality is obtained by duality.

(d) The claim follows the facts that $P=1-Q$ and $Q\ge S_j$ for $j=1,2,3$.
\end{proof}

One can show that the operator $I_1(\kappa)$ appearing in \eqref{eq_JN} does not give
any contribution. Indeed, we know from \cite[Eq.~(6.27)]{JN01} that
\begin{equation}\label{eq_I_0}
(\M(\kappa)+S_1)^{-1}
=g(\kappa)^{-1}I_0(\kappa)+QD_0(\kappa)Q
\end{equation}
with $I_0(\kappa)$ and $D_0(\kappa)$ operators having limits in $\B(\H)$ as
$\lambda\searrow0$. Since $g(\kappa)^{-1}=O(\eta)$ as $\lambda\searrow0$, it follows
that
$$
\lim_{\lambda\searrow0}(\M(\kappa)+S_1)^{-1}=QD_0(0)Q
\quad\hbox{with}\quad
D_0(0):=\lim_{\lambda\searrow0}D_0(\kappa).
$$
Therefore, \eqref{eq_I_1} and Lemma \ref{lemma_Q}(a) imply that
\begin{equation}\label{eq_partial_1}
\lim_{\lambda\searrow0}
\big(\gamma_0+\sqrt\lambda\;\!\gamma_1+\gamma_2(\lambda)\big)
vI_1(\kappa)v
\big(\gamma_0^*+\sqrt\lambda\;\!\gamma_1^*+\gamma_2(\lambda)^*\big)
=\lim_{\lambda\searrow0}\gamma_0vQD_0(0)Qv\gamma_0^*
=0,
\end{equation}
showing that the first term in \eqref{eq_JN} does not lead to any contribution in
\eqref{eq_to_prove}.

We can now turn our attention to the second term in \eqref{eq_JN}. For its analysis,
we introduce the operator
$$
M_0(\kappa):=\M(\kappa)+\tfrac1{2\pi\eta}\|v\|^2_\H\;\!P,
$$
and we note from \cite[Eq.~(6.28)]{JN01} that $M_0(\kappa)=M_{0,0}+O(\lambda/\eta)$
as $\lambda\searrow0$, with $M_{0,0}\in\B(\H)$ self-adjoint.

\begin{Lemma}\label{lemma_comm}
One has $[(\M(\kappa)+S_1)^{-1},S_1]=O(\eta)$ as $\lambda\searrow0$.
\end{Lemma}

\begin{proof}
Using \eqref{eq_I_0}, the fact that $D_0(\kappa)=\big(Q(M_0(\kappa)+S_1)Q\big)^{-1}$
\cite[Eq.~(6.29)]{JN01}, and the fact that $[Q,S_1]=0$, we get as $\lambda\searrow0$
\begin{align*}
[(\M(\kappa)+S_1)^{-1},S_1]
&=[QD_0(\kappa)Q,S_1]+O(\eta)\\
&=Q\big[\big(Q(M_0(\kappa)+S_1)Q\big)^{-1},S_1\big]Q+O(\eta)\\
&=QD_0(\kappa)\big[S_1,Q(M_0(\kappa)+S_1)Q\big]D_0(\kappa)Q+O(\eta).
\end{align*}
Since $M_0(\kappa)=M_{0,0}+O(\lambda/\eta)$ and $S_1$ is the orthogonal projection on
$\Ker(QM_{0,0}Q)$ (see \cite[Thm.~6.2(i)]{JN01}), we infer that
$$
[(\M(\kappa)+S_1)^{-1},S_1]
=QD_0(\kappa)\;\!O(\lambda/\eta)\;\!D_0(\kappa)Q+O(\eta)
=O(\eta)
\quad\hbox{as $\lambda\searrow0$,}
$$
as desired.
\end{proof}

The equation \eqref{eq_I_2} and Lemma \ref{lemma_comm} imply that the second term in
\eqref{eq_JN} satisfies as $\lambda\searrow0$
\begin{align*}
g(\kappa)I_2(\kappa)
&=g(\kappa)\big(S_1(\M(\kappa)+S_1)^{-1}+O(\eta)\big)\\
&\quad\cdot(M_1(\kappa)+S_2)^{-1}
\big((\M(\kappa)+S_1)^{-1}S_1+O(\eta)\big).
\end{align*}
Thus, taking into account Lemma \ref{lemma_Q}(b) and the fact that
$g(\kappa)=O(\eta^{-1})$ as $\lambda\searrow0$, one gets
\begin{align}
&\lim_{\lambda\searrow0}
\big(\gamma_0+\sqrt\lambda\;\!\gamma_1+\lambda\gamma_2(\lambda)\big)v
g(\kappa)I_2(\kappa)v
\big(\gamma_0^*+\sqrt\lambda\;\!\gamma_1^*+\lambda\gamma_2(\lambda)^*\big)\nonumber\\
&=\lim_{\lambda\searrow0}g(\kappa)\big(\gamma_0+\sqrt\lambda\;\!\gamma_1
+\lambda\gamma_2(\lambda)\big)v
\big(S_1(\M(\kappa)+S_1)^{-1}+O(\eta)\big)\nonumber\\
&\quad\cdot(M_1(\kappa)+S_2)^{-1}
\big((\M(\kappa)+S_1)^{-1}S_1+O(\eta)\big)v
\big(\gamma_0^*+\sqrt\lambda\;\!\gamma_1^*+\lambda\gamma_2(\lambda)^*\big)\nonumber\\
&=0,\label{eq_partial_2}
\end{align}
meaning that the second term in \eqref{eq_JN} does not lead to any contribution in
\eqref{eq_to_prove}.

Let us now consider the third term in \eqref{eq_JN}. Since all factors in
$I_3(\kappa)$ have limits as $\lambda\searrow0$ and $g(\kappa)=O(\eta^{-1})$ as
$\lambda\searrow0$, that term behaves at worst like $O(\lambda^{-1})$ as
$\lambda\searrow0$. Therefore, the terms in
$$
\tfrac{g(\kappa)\eta}\lambda
\big(\gamma_0+\sqrt\lambda\;\!\gamma_1+\lambda\gamma_2(\lambda)\big)vI_3(\kappa)v
\big(\gamma_0^*+\sqrt\lambda\;\!\gamma_1^*+\lambda\gamma_2(\lambda)^*\big)
$$
that do not manifestly vanish in the limit $\lambda\searrow0$ are:
\begin{enumerate}
\item[(i)]
$\tfrac{g(\kappa)\eta}\lambda\;\!\gamma_0vI_3(\kappa)v\gamma_0^*$,
\item[(ii)]
$
\tfrac{g(\kappa)\eta}{\sqrt\lambda}\;\!
\big(\gamma_0vI_3(\kappa)v\gamma_1^*
+\gamma_1vI_3(\kappa)v\gamma_0^*\big)
$,
\item[(iii)] $g(\kappa)\eta\;\!\gamma_1vI_3(\kappa)v\gamma_1^*$,
\item[(iv)]
$
g(\kappa)\eta\big(\gamma_0vI_3(\kappa)v\gamma_2^*
+\gamma_2vI_3(\kappa)v\gamma_0^*\big)
$.
\end{enumerate}

In order to study these terms, some preparatory lemmas are necessary, starting with
one on commutators:

\begin{Lemma}\label{lemma_comm_2}
For $j=2,3$, one has $[(M_1(\kappa)+S_2)^{-1},S_j]=O(\lambda/\eta^2)$ in $\B(S_1\H)$
as $\lambda\searrow0$.
\end{Lemma}

\begin{proof}
We show the claim for $j=2$, since the case $j=3$ is similar. We know from
\cite[Eq.~(6.31)]{JN01} that the operator $M_1(\kappa)$ defined in $S_1\H$ satisfies
as $\lambda\searrow0$ the expansion $M_1(\kappa)=M_{1;0,0}+O(\lambda/\eta^2)$ with
$M_{1;0,0}:=S_1M_{0,0}PM_{0,0}S_1$. It follows that
\begin{align*}
[(M_1(\kappa)+S_2)^{-1},S_2]
&=(M_1(\kappa)+S_2)^{-1}[S_2,M_1(\kappa)+S_2](M_1(\kappa)+S_2)^{-1}\\
&=(M_1(\kappa)+S_2)^{-1}[S_2,M_{1;0,0}](M_1(\kappa)+S_2)^{-1}+O(\lambda/\eta^2).
\end{align*}
Since $S_2$ is the projection on the kernel of $M_{1;0,0}$ (see
\cite[Thm.~6.2(i)]{JN01}), this implies the claim.
\end{proof}

In the next proposition, we deal with the simplest limits above, the ones of (iii) and
(iv).

\begin{Proposition}\label{prop_iii_iv}
One has
\begin{equation}\label{eq_lim_iii}
\lim_{\lambda\searrow0}
g(\kappa)\eta\;\!\gamma_1vI_3(\kappa)v\gamma_1^*=0
\end{equation}
and
\begin{equation}\label{eq_lim_iv}
\lim_{\lambda\searrow0}
g(\kappa)\eta\big(\gamma_0vI_3(\kappa)v\gamma_2^*
+\gamma_2vI_3(\kappa)v\gamma_0^*\big)=0.
\end{equation}
\end{Proposition}

\begin{proof}
Since $\lim_{\lambda\searrow0}g(\kappa)\eta$ exists, and since all the factors in
$I_3(\kappa)$ have limits as $\lambda\searrow0$, one can factorise the first limit as
$$
\lim_{\lambda\searrow0}
g(\kappa)\eta\;\!\gamma_1vI_3(\kappa)v\gamma_1^*
=\lim_{\lambda\searrow0}g(\kappa)\eta
\cdot\lim_{\lambda\searrow0}\gamma_1vI_3(\kappa)v\gamma_1^*.
$$
So, it is sufficient to show that
$\lim_{\lambda\searrow0}\gamma_1vI_3(\kappa)v\gamma_1^*=0$ to prove
\eqref{eq_lim_iii}. Now, we have $S_1\ge S_2\ge S_3$, $m(\kappa)^{-1}=O(\eta)$ as
$\lambda\searrow0$, and also
\begin{align}
&\lim_{\lambda\searrow0}(\M(\kappa)+S_1)^{-1}S_1
=S_1=\lim_{\lambda\searrow0}S_1(\M(\kappa)+S_1)^{-1},
\label{eq_lim_M}\\
&\lim_{\lambda\searrow0}(M_1(\kappa)+S_2)^{-1}S_2
=S_2=\lim_{\lambda\searrow0}S_2(M_1(\kappa)+S_2)^{-1},\nonumber
\end{align}
due to Lemmas \ref{lemma_comm}-\ref{lemma_comm_2} (and their proofs). Therefore, we
get from \eqref{eq_I_3}
\begin{align*}
\lim_{\lambda\searrow0}\gamma_1vI_3(\kappa)v\gamma_1^*
&=\lim_{\lambda\searrow0}\gamma_1v
\big(-T_3m(\kappa)^{-1}b(\kappa)d(\kappa)^{-1}S_3
-S_3d(\kappa)^{-1}c(\kappa)m(\kappa)^{-1}T_3\\
&\quad+S_3d(\kappa)^{-1}c(\kappa)m(\kappa)^{-1}
b(\kappa)d(\kappa)^{-1}S_3+S_3d(\kappa)^{-1}S_3\big)v\gamma_1^*,
\end{align*}
and thus obtain that $\lim_{\lambda\searrow0}\gamma_1vI_3(\kappa)v\gamma_1^*=0$ thanks
to Lemma \ref{lemma_Q}(c).

Similarly, in order to prove \eqref{eq_lim_iv} it is sufficient to show that
$\lim_{\lambda\searrow0}\gamma_0vI_3(\kappa)v\gamma_2^*=0$ (or that
$\lim_{\lambda\searrow0}\gamma_2vI_3(\kappa)v\gamma_0^*=0$, this is similar). In this
case, we get
\begin{align*}
\lim_{\lambda\searrow0}\gamma_0vI_3(\kappa)v\gamma_2^*
&=\lim_{\lambda\searrow0}\gamma_0v\big(-T_3m(\kappa)^{-1}b(\kappa)d(\kappa)^{-1}S_3
-S_3d(\kappa)^{-1}c(\kappa)m(\kappa)^{-1}T_3\\
&\quad+S_3d(\kappa)^{-1}c(\kappa)m(\kappa)^{-1}
b(\kappa)d(\kappa)^{-1}S_3
+S_3d(\kappa)^{-1}S_3\big)v\gamma_2^*,
\end{align*}
and thus obtain that $\lim_{\lambda\searrow0}\gamma_0vI_3(\kappa)v\gamma_2^*=0$ thanks
to Lemma \ref{lemma_Q}(b).
\end{proof}

For the remaining terms (i) and (ii), we need two more lemmas.

\begin{Lemma}\label{lemma_M_00}~
\begin{enumerate}
\item[(a)] One has $PM_{0,0}S_2=0$.
\item[(b)] One has $\gamma_0vM_{0,0}S_3=0=S_3M_{0,0}v\gamma_0^*$.
\item[(c)] One has as $\lambda\searrow0$
$$
PM_0(\kappa)QD_0(\kappa)S_2=O(\lambda/\eta)
\quad\hbox{and}\quad
S_2D_0(\kappa)QM_0(\kappa)P=O(\lambda/\eta).
$$
\end{enumerate}
\end{Lemma}

\begin{proof}
(a) Let $f\in\H$ and $g:=S_2f$. Then, we have by definition of $S_2$ the inclusion
$g\in\Ker(M_{1;0,0})=\Ker(S_1 M_{0,0}PM_{0,0}S_1)$. Therefore, we get the equalities
$$
0=\big\langle g,S_1M_{0,0}PM_{0,0}S_1g\big\rangle_\H
=\|PM_{0,0}S_1g\|^2_\H
=\|PM_{0,0}S_2f\|^2_\H,
$$
which imply the claim.

(b) We have for any $f\in\H$
$$
2^{3/2}\pi\gamma_0vM_{0,0}S_3f
=\int_{\R^2}\d x\,v(x)(M_{0,0}S_3f)(x)
=\big\langle v,M_{0,0}S_3f\big\rangle_\H
=0,
$$
with the last equality following from \cite[Eq.~(6.100)]{JN01}. This shows the first
equality. The second equality is then obtained by duality.

(c) A successive application of Equations (6.58), (6.69) and (6.56) of \cite{JN01}
gives as $\lambda\searrow0$
\begin{align*}
PM_0(\kappa)QD_0(\kappa)S_2
&=PM_{0,0}QD_0(\kappa)S_2+O(\lambda/\eta)\\
&=PM_{0,0}Q(QM_{0,0}Q+S_1)^{-1}S_2+O(\lambda/\eta)\\
&=PM_{0,0}QS_2+O(\lambda/\eta).
\end{align*}
Since $PM_{0,0}QS_2=PM_{0,0}S_2=0$ by point (a), we obtain the first equality. The
second equality is obtained similarly.
\end{proof}

\begin{Lemma}\label{lemma_gamma_0}
One has as $\lambda\searrow0$
\begin{align}
&\gamma_0v(\M(\kappa)+S_1)^{-1}S_1
(M_1(\kappa)+S_2)^{-1}S_2
=O(\lambda/\eta),\label{eq_bound_1}\\
&S_2(M_1(\kappa)+S_2)^{-1}S_1
(\M(\kappa)+S_1)^{-1}v\gamma_0^*
=O(\lambda/\eta).\label{eq_bound_2}
\end{align}
\end{Lemma}

\begin{proof}
Using successively \cite[Eq.~(6.27)]{JN01}, Lemmas \ref{lemma_Q}(a) \&
\ref{lemma_Q}(d), Lemma \ref{lemma_comm_2}, Lemma \ref{lemma_M_00}(c) and the estimate
$g(\kappa)^{-1}=O(\eta)$, we obtain the equalities
\begin{align*}
&\gamma_0v(\M(\kappa)+S_1)^{-1}S_1(M_1(\kappa)+S_2)^{-1}S_2\\
&=-g(\kappa)^{-1}\gamma_0vPM_0(\kappa)QD_0(\kappa)S_1
(M_1(\kappa)+S_2)^{-1}S_2\\
&=-g(\kappa)^{-1}\gamma_0vPM_0(\kappa)QD_0(\kappa)
\big(S_2(M_1(\kappa)+S_2)^{-1}+O(\lambda/\eta^2)\big)\\
&=-g(\kappa)^{-1}\gamma_0v\big(O(\lambda/\eta)+O(\lambda/\eta^2)\big)\\
&=O(\lambda/\eta).
\end{align*}
This proves \eqref{eq_bound_1}. The equality \eqref{eq_bound_2} can be shown
similarly.
\end{proof}

We are now ready to deal with the terms (i) and (ii):

\begin{Proposition}\label{prop_i_ii}
One has
$$
\lim_{\lambda\searrow0}
\tfrac{g(\kappa)\eta}\lambda\;\!\gamma_0vI_3(\kappa)v\gamma_0^*=0
$$
and
$$
\lim_{\lambda\searrow0}\tfrac{g(\kappa)\eta}{\sqrt\lambda}\;\!
\big(\gamma_0vI_3(\kappa)v\gamma_1^*
+\gamma_1vI_3(\kappa)v\gamma_0^*\big)=0.
$$
\end{Proposition}

\begin{proof}
The operator $I_3(\kappa)$ is of the form (see \eqref{eq_I_3})
$$
(\M(\kappa)+S_1)^{-1}S_1(M_1(\kappa)+S_2)^{-1}S_2\;\!(\dots)\;\!
S_2(M_1(\kappa)+S_2)^{-1}S_1(\M(\kappa)+S_1)^{-1},
$$
with $(\dots)$ having a limit as $\lambda\searrow0$. Therefore, the two claims follow
from an application of Lemma \ref{lemma_gamma_0}.
\end{proof}

We can finally give the proof of the main result of this section:

\begin{proof}[Proof of Theorem \ref{thm_S_matrix}]
The proof reduces to gathering the information contained in Equations
\eqref{eq_partial_1}-\eqref{eq_partial_2} and Propositions \ref{prop_iii_iv} \&
\ref{prop_i_ii}. As a result, we obtain that all the contributions appearing in the
formula \eqref{eq_S_matrix} for the $S$-matrix $S(\lambda)$ vanish in the limit
$\lambda\searrow0$, except the first term $1_\hs$.
\end{proof}

\section{Wave operators}\label{sec_wave}
\setcounter{equation}{0}

If the potential $V$ satisfies \eqref{eq_cond_V} with $\rho>1$, then the stationary
wave operators and the strong wave operators exist and coincide (see
\cite[Thm.~5.3.6]{Yaf92}). So, starting from the formula for the stationary wave
operators \cite[Eq.~2.7.5]{Yaf92} and taking into account the resolvent equation
written in the symmetrised form \cite[Eq.~4.3]{JN01}, one obtains for suitable
$\varphi,\psi\in\Hrond:$
\begin{align}
&\big\langle\F_0\big(W_\pm-1\big)\F_0^*\varphi,
\psi\big\rangle_\Hrond\label{start}\\
&=-\int_\R\d\lambda\,\lim_{\varepsilon\searrow0}\int_0^\infty\d\mu\,
\big\langle \F_0(\mu)v\big(u+vR_0(\lambda\mp i\varepsilon)v\big)^{-1}v\F_0^*
\delta_\varepsilon(L-\lambda)\varphi,(\mu-\lambda\mp i\varepsilon)^{-1}
\psi(\mu)\big\rangle_\hs\nonumber
\end{align}
where
$$
\delta_\varepsilon(L-\lambda)
:=\tfrac\varepsilon\pi(L-\lambda+i\varepsilon)^{-1}(L-\lambda-i\varepsilon)^{-1}.
$$

In order to exchange the limit $\varepsilon\searrow0$ and the integral over $\mu$, one
needs to collect some preparatory results. The first of them is a lemma on the
operator $\F_0^*\delta_\varepsilon(L-\lambda)$ appearing in \eqref{start}. We use the
notation $C_{\rm c}(\R_+;\G)$ for the set of compactly supported continuous functions
from $\R_+$ to some Hilbert space $\G$.

\begin{Lemma}[Lemma 2.3 of \cite{RT13}]\label{lemma_limit}
For $s\ge0$, $t>1$, $\lambda\in\R_+$ and $\varphi\in C_{\rm c}(\R_+;\hs)$, one has
$
\lim_{\varepsilon\searrow0}\F_0^*\;\!\delta_\varepsilon(L-\lambda)\varphi
=\F_0(\lambda)^*\varphi(\lambda)
$
in $\H^{-s}_{-t}$.
\end{Lemma}

The next task is to analyse the function
$\lambda\mapsto(u+vR_0(\lambda\mp i0)v)^{-1}v\F_0(\lambda)^*$, which will appear in
\eqref{start} once the limit $\varepsilon\searrow0$ is taken. This is the content of
the next section.

\subsection{The asymmetric term}\label{sec_asym}

In this section, we determine the behaviour of the function
$\lambda\mapsto(u+vR_0(\lambda\mp i0)v)^{-1}v\F_0(\lambda)^*$ as $\lambda\searrow0$.
The main difference with respect to the analysis conducted in Section
\ref{sec_scattering} is the absence of a factor $\F_0(\lambda)v$ on the left of the
operator $(u+vR_0(\lambda\mp i0)v)^{-1}v\F_0(\lambda)^*$.	For this reason, we call
this operator the \emph{asymmetric term}. As in Section \ref{sec_scattering}, we
assume that $V$ satisfies \eqref{eq_cond_V} with $\rho>11$ so that both the expansions
\eqref{eq_exp_F_0} for $\F_0(\lambda)$ and \eqref{eq_JN} for
$(u+vR_0(-\kappa^2)v)^{-1}$ hold. The main result of this section is presented in
Theorem \ref{thm_asym}, the lemmas and propositions coming before are preparation for
it.

In our first lemma, we determine the behaviour of the simplest terms appearing in the
expansion of $(u+vR_0(\lambda\mp i0)v)^{-1}v\F_0(\lambda)^*$. We keep using the
shorthand notations $\kappa=-i\sqrt\lambda$ and $\eta=1/\ln(\kappa)$ introduced in
\eqref{eq_convention}.

\begin{Lemma}
One has as $\lambda\searrow0$
\begin{enumerate}
\item[(a)] $I_1(\kappa)v\F_0(\lambda)^*=O(\eta)$,
\item[(b)] $g(\kappa)I_2(\kappa)v\F_0(\lambda)^*=O(1)$,
\item[(c)] $\tfrac{g(\kappa)\eta}\lambda\;\!I_3(\kappa)v\gamma_2(\lambda)^*=O(1)$.
\end{enumerate}
\end{Lemma}

\begin{proof}
Using \eqref{eq_exp_F_0}, \eqref{eq_I_1}, \eqref{eq_I_0} and Lemma \ref{lemma_Q}(a),
we obtain the first claim:
$$
I_1(\kappa)v\F_0(\lambda)^*
=\big(QD_0(\kappa)Q+O(\eta)\big)v
\big(\gamma_0^*+\sqrt\lambda\;\!\gamma_1^*+O(\lambda)\big)
=O(\eta)\quad\hbox{as $\lambda\searrow0$.}
$$
For the second claim, we note from \eqref{eq_exp_F_0}, \eqref{eq_I_2} and Lemma
\ref{lemma_comm} that as $\lambda\searrow0$
\begin{align*}
g(\kappa)I_2(\kappa)v\F_0(\lambda)^*
&=g(\kappa)(\M(\kappa)+S_1)^{-1}S_1
(M_1(\kappa)+S_2)^{-1}\\
&\quad\cdot\big((\M(\kappa)+S_1)^{-1}S_1+O(\eta)\big)v
\big(\gamma_0^*+\sqrt\lambda\;\!\gamma_1^*+O(\lambda)\big).
\end{align*}
Since $g(\kappa)=O(\eta^{-1})$ as $\lambda\searrow0$ and $S_1v\gamma_0^*=0$ due to
Lemma \ref{lemma_Q}(b), we infer that
$$
g(\kappa)I_2(\kappa)v\F_0(\lambda)^*=O(1)
\quad\hbox{as $\lambda\searrow0$.}
$$
Finally, the third claim follows from the facts that $g(\kappa)=O(\eta^{-1})$ and
$\gamma_2(\lambda)=O(\lambda)$ as $\lambda\searrow0$.
\end{proof}

We can now focus on the remaining two terms of the expansion of
$(u+vR_0(\lambda\mp i0)v)^{-1}v\F_0(\lambda)^*:$
\begin{equation}\label{eq_2_terms}
\tfrac{g(\kappa)\eta}\lambda\;\!I_3(\kappa)v\gamma_0^*
\quad\hbox{and}\quad
\tfrac{g(\kappa)\eta}{\sqrt\lambda}\;\!I_3(\kappa)v\gamma_1^*.
\end{equation}
For this, we first observe from \eqref{eq_I_3} that $I_3(\kappa)$ can be rewritten as
\begin{equation}\label{eq_sum}
I_3(\kappa)
=\big(B_S(\kappa)S_3+B_T(\kappa)T_3\big)S_2
(M_1(\kappa)+S_2)^{-1}S_1(\M(\kappa)+S_1)^{-1}
\end{equation}
with
\begin{align*}
B_S(\kappa)
&:=(\M(\kappa)+S_1)^{-1}S_1(M_1(\kappa)+S_2)^{-1}
S_2\big(-T_3m(\kappa)^{-1}b(\kappa)d(\kappa)^{-1}\\
&\quad+S_3d(\kappa)^{-1}c(\kappa)m(\kappa)^{-1}
b(\kappa)d(\kappa)^{-1}+S_3d(\kappa)^{-1}\big)
\end{align*}
and
$$
B_T(\kappa)
:=(\M(\kappa)+S_1)^{-1}S_1(M_1(\kappa)+S_2)^{-1}S_2
\big(T_3m(\kappa)^{-1}-S_3d(\kappa)^{-1}c(\kappa)
m(\kappa)^{-1}\big).
$$
Then, we first consider a part of the second operator in \eqref{eq_2_terms}:

\begin{Lemma}\label{lemma_B_S}
One has
$$
B_S(\kappa)S_3(M_1(\kappa)+S_2)^{-1}S_1
(\M(\kappa)+S_1)^{-1}v\gamma_1^*
=O(\lambda/\eta^2)\quad\hbox{as $\lambda\searrow0$.}
$$
\end{Lemma}

\begin{proof}
Using successively Lemma \ref{lemma_comm_2}, \cite[Eq.~(6.27)]{JN01}, the fact that
$g(\kappa)^{-1}=O(\eta)$ as $\lambda\searrow0$, Lemma \ref{lemma_M_00}(c), the fact
that $M_0(\kappa)$ has a limit as $\lambda\searrow0$, and \cite[Eq.~(6.69)]{JN01}, we
get as $\lambda\searrow0$
\begin{align*}
&S_3(M_1(\kappa)+S_2)^{-1}S_1
(\M(\kappa)+S_1)^{-1}v\gamma_1^*\\
&=(M_1(\kappa)+S_2)^{-1}S_3
(\M(\kappa)+S_1)^{-1}v\gamma_1^*+O(\lambda/\eta^2)\\
&=(M_1(\kappa)+S_2)^{-1}S_3\big\{g(\kappa)^{-1}
\big(-QD_0(\kappa)QM_0(\kappa)P\\
&\quad+Q D_0(\kappa)QM_0(\kappa)PM_0(\kappa)Q
D_0(\kappa)Q\big)+QD_0(\kappa)Q\big\}v\gamma_1^*+O(\lambda/\eta^2)\\
&=(M_1(\kappa)+S_2)^{-1}S_3D_0(0)Qv\gamma_1^*+O(\lambda/\eta^2).
\end{align*}
Since $S_3D_0(0)Q=S_3$ due to \cite[Eq.~(6.56)]{JN01}, we infer from Lemma
\ref{lemma_Q}(c) that
$$
S_3(M_1(\kappa)+S_2)^{-1}S_1
(\M(\kappa)+S_1)^{-1}v\gamma_1^*
=O(\lambda/\eta^2)\quad\hbox{as $\lambda\searrow0$.}
$$
To conclude, it only remains to observe that $B_S(\kappa)$ has a limit as
$\lambda\searrow0$.
\end{proof}

For the second term in \eqref{eq_sum}, we cannot get a similar estimate since
$T_3=S_2-S_3$, with the projection $S_2$ not leading to many simplifications. In this
case, we only get:

\begin{Lemma}\label{lemma_B_T}
One has
\begin{align*}
&B_T(\kappa)T_3(M_1(\kappa)+S_2)^{-1}S_1
(\M(\kappa)+S_1)^{-1}v\gamma_1^*\\
&=\big(T_3m(\kappa)^{-1}-S_3d(\kappa)c(\kappa)
m(\kappa)^{-1}\big)T_3v\gamma_1^*+O(\lambda/\eta^2)
\quad\hbox{as $\lambda\searrow0$.}
\end{align*}
\end{Lemma}

\begin{proof}
A calculation as in the proof of Lemma \ref{lemma_B_S} gives as $\lambda\searrow0$
\begin{align*}
T_3(M_1(\kappa)+S_2)^{-1}S_1(\M(\kappa)+S_1)^{-1}v\gamma_1^*
&=(M_1(\kappa)+S_2)^{-1}T_3
(\M(\kappa)+S_1)^{-1}v\gamma_1^*+O(\lambda/\eta^2)\\
&=(M_1(\kappa)+S_2)^{-1}T_3D_0(0)Qv\gamma_1^*+O(\lambda/\eta^2)\\
&=(M_1(\kappa)+S_2)^{-1}T_3v\gamma_1^*+O(\lambda/\eta^2).
\end{align*}
Furthermore, we know from the proof of Lemma \ref{lemma_comm_2} that
$M_1(\kappa)=M_{1;0,0}+O(\lambda/\eta^2)$ in $\B(S_1\H)$ as $\lambda\searrow0$ and
that $S_2$ is the projection on the kernel of $M_{1;0,0}$. So,
\begin{equation}\label{eq_exp_M_1}
(M_1(\kappa)+S_2)^{-1}S_2=S_2+O(\lambda/\eta^2)
\quad\hbox{as $\lambda\searrow0$,}
\end{equation}
and we obtain
\begin{equation}\label{eq_first_bound}
T_3(M_1(\kappa)+S_2)^{-1}S_1
(\M(\kappa)+S_1)^{-1}v\gamma_1^*
=T_3v\gamma_1^*+O(\lambda/\eta^2)\quad\hbox{as $\lambda\searrow0$.}
\end{equation}
Similarly, a calculation as in the proof of Lemma \ref{lemma_B_S} gives as
$\lambda\searrow0$
\begin{align}
B_T(\kappa)
&=(\M(\kappa)+S_1)^{-1}S_1(M_1(\kappa)+S_2)^{-1}S_2
\big(T_3m(\kappa)^{-1}-S_3d(\kappa)^{-1}c(\kappa)
m(\kappa)^{-1}\big)\nonumber\\
&=(\M(\kappa)+S_1)^{-1}\big(T_3m(\kappa)^{-1}
-S_3d(\kappa)^{-1}c(\kappa)m(\kappa)^{-1}\big)
+O(\lambda/\eta^2)\nonumber\\
&=QD_0(0)\big(T_3m(\kappa)^{-1}-S_3d(\kappa)^{-1}c(\kappa)
m(\kappa)^{-1}\big)+O(\lambda/\eta^2)\nonumber\\
&=T_3m(\kappa)^{-1}-S_3d(\kappa)^{-1}c(\kappa)
m(\kappa)^{-1}+O(\lambda/\eta^2).\label{eq_second_bound}
\end{align}
Thus, one infers the claim by combining \eqref{eq_first_bound} and
\eqref{eq_second_bound}.
\end{proof}

Using the previous two lemmas, we get the following estimate for the second term in
\eqref{eq_2_terms}:

\begin{Proposition}
One has as $\lambda\searrow0$
\begin{align*}
\tfrac{g(\kappa)\eta}{\sqrt\lambda}\;\!I_3(\kappa)v\gamma_1^*
=\tfrac{g(\kappa)\eta}{\sqrt\lambda}\big(T_3-S_3d(\kappa)^{-1}c(\kappa)\big)
m(\kappa)^{-1}T_3v\gamma_1^*+O(\sqrt\lambda/\eta^2).
\end{align*}
\end{Proposition}

And for the first operator in \eqref{eq_2_terms}, we get:

\begin{Lemma}
One has as $\lambda\searrow0$
$$
\tfrac{g(\kappa)\eta}\lambda\;\!I_3(\kappa)v\gamma_0^*
=\tfrac1\eta S_3\;\!O(1)+O(1).
$$
\end{Lemma}

\begin{proof}
Using \cite[Eq.~(6.27)]{JN01}, Lemmas \ref{lemma_Q}(a) \& (d), Lemma
\ref{lemma_M_00}(c), and the fact that $g(\kappa)^{-1}=O(\eta)$ as
$\lambda\searrow0$, we obtain as $\lambda\searrow0$
\begin{align*}
&S_2(M_1(\kappa)+S_2)^{-1}S_1
(\M(\kappa)+S_1)^{-1}v\gamma_0^*\\
&=-g(\kappa)^{-1}S_2(M_1(\kappa)+S_2)^{-1}
S_1D_0(\kappa)QM_0(\kappa)Pv\gamma_0^*\\
&=-g(\kappa)^{-1}[S_2,(M_1(\kappa)+S_2)^{-1}]
S_1D_0(\kappa)QM_0(\kappa)Pv\gamma_0^*+O(\lambda).
\end{align*}
Since $\tfrac{\eta^2}\lambda[S_2,(M_1(\kappa)+S_2)^{-1}]$ has a limit in
$\B(S_1\H)$ as $\lambda\searrow0$ due to Lemma \ref{lemma_comm_2} and since
$\tfrac1\eta\;\!m(\kappa)^{-1}$ has a limit in $\B(\H)$ as $\lambda\searrow0$, it
follows from what precedes and \eqref{eq_I_3} that
\begin{align*}
\tfrac{g(\kappa)\eta}\lambda\;\!I_3(\kappa)v\gamma_0^*
&=-(\M(\kappa)+S_1)^{-1}S_1(M_1(\kappa)+S_2)^{-1}S_2
\big(T_3m(\kappa)^{-1}T_3-T_3m(\kappa)^{-1}b(\kappa)d(\kappa)^{-1}S_3\\
&\quad-S_3d(\kappa)^{-1}c(\kappa)m(\kappa)^{-1}T_3
+S_3d(\kappa)^{-1}c(\kappa)m(\kappa)^{-1}b(\kappa)d(\kappa)^{-1}S_3
+S_3d(\kappa)^{-1}S_3\big)\\
&\quad\cdot\tfrac\eta\lambda[S_2,(M_1(\kappa)+S_2)^{-1}]
S_1D_0(\kappa)QM_0(\kappa)Pv\gamma_0^*+O(1)\nonumber\\
&=-\tfrac1\eta(\M(\kappa)+S_1)^{-1}S_1(M_1(\kappa)+S_2)^{-1}S_3d(\kappa)^{-1}S_3
\tfrac{\eta^2}\lambda[S_2,(M_1(\kappa)+S_2)^{-1}]\nonumber\\
&\quad\cdot S_1D_0(\kappa)QM_0(\kappa)Pv\gamma_0^*+O(1).
\end{align*}
Now, \eqref{eq_exp_M_1} implies that
$(M_1(\kappa)+S_2)^{-1}S_3=S_3+O(\lambda/\eta^2)$ as $\lambda\searrow0$, and
\eqref{eq_I_0} \& \eqref{eq_lim_M} imply that $(\M(\kappa)+S_1)^{-1}S_1=S_1+O(\eta)$
as $\lambda\searrow0$. So, we have
$$
(\M(\kappa)+S_1)^{-1}S_1(M_1(\kappa)+S_2)^{-1}S_3
=S_3+O(\eta)\quad\hbox{as $\lambda\searrow0$.}
$$
It follows that for $\lambda\searrow0$
\begin{align*}
&\tfrac{g(\kappa)\eta}\lambda\;\!I_3(\kappa)v\gamma_0^*
=-\tfrac1\eta S_3d(\kappa)^{-1}S_3\tfrac{\eta^2}\lambda
[S_2,(M_1(\kappa)+S_2)^{-1}]S_1D_0(\kappa)Q
M_0(\kappa)Pv\gamma_0^*+O(1).
\end{align*}
Since the operators $d(\kappa)^{-1}$,
$\tfrac{\eta^2}\lambda[S_2,(M_1(\kappa)+S_2)^{-1}]$, $D_0(\kappa)$, and $M_0(\kappa)$
have a limit as $\lambda\searrow0$, we finally obtain that
$$
\tfrac{g(\kappa)\eta}\lambda\;\!I_3(\kappa)v\gamma_0^*
=\tfrac1\eta S_3\;\!O(1)+O(1)\quad\hbox{as $\lambda\searrow0$,}
$$
as desired.
\end{proof}

By collecting what precedes, we can finally get a description of the behaviour of the
asymmetric term:

\begin{Theorem}\label{thm_asym}
If $V$ satisfies \eqref{eq_cond_V} with $\rho>11$, then one has as $\lambda\searrow0$
$$
\big(u+vR_0(\lambda\mp i0)v\big)^{-1}v\F_0(\lambda)^*
=\tfrac{g(\kappa)\eta}{\sqrt\lambda}\big(T_3-S_3d(\kappa)^{-1}c(\kappa)\big)
m(\kappa)^{-1}T_3v\gamma_1^*+\tfrac1\eta S_3\;\!O(1)+O(1).
$$
\end{Theorem}

Note that since $g(\kappa)=O(\eta^{-1})$ and $m(\kappa)^{-1}=O(\eta)$ as
$\lambda\searrow0$, the first term behaves as $O(\eta/\sqrt{\lambda})$ in the limit
$\lambda\searrow0$. This singular behaviour is due to p-resonances since the term
vanishes when $T_3=0$ (see the discussion in Section \ref{section_intro}). On the
other hand, the singular term $\tfrac1\eta S_3\;\!O(1)$ is associated with $0$-energy
bound states, since it vanishes when $S_3=0$.

\subsection{Explicit formula for the wave operators}

Our final objective is the derivation of an explicit formula for the wave operators
$W_\pm$. A formula of this type has already been obtained in \cite{RT13_2}, but only
in the generic case (when the $0$-energy is assumed to be neither an eigenvalue nor a
resonance). We start by recalling continuity properties of the function
$\lambda\mapsto\F_0(\lambda):$

\begin{Lemma}[Continuity properties of $\F_0(\lambda)$]\label{lemma_F_0}~
\begin{enumerate}
\item[(a)] For any $s\ge0$ and $t>1/2$, the function
$\R_+\ni\lambda\mapsto\langle\lambda\rangle^{1/4}\F_0(\lambda)\in\B(\H^s_t,\hs)$ is
continuous and bounded.
\item[(b)] Let $s>-1/2$ and $t>1$. Then, $\F_0(\lambda)\in\K(\H^s_t,\hs)$ for each
$\lambda\in\R_+$, and the function
$\R_+\ni\lambda\mapsto\F_0(\lambda)\in\K(\H^s_t,\hs)$ is continuous, admits a limit
as $\lambda\searrow0$, and vanishes as $\lambda\to\infty$.
\end{enumerate}
\end{Lemma}

\begin{proof}
Point (a) follows from the properties of $\F_0(\cdot)$ presented after \eqref{eq_diag}
and the estimate \cite[Thm.~1.1.4]{Yaf10}. The proof of point (b) is analogous to that
of \cite[Lemma~2.2]{RT13}.
\end{proof}

In the next lemma, we define and determine continuity properties of two
operator-valued functions which play a key role in the sequel. As in Section
\ref{sec_asym}, we assume that $V$ satisfies \eqref{eq_cond_V} with $\rho>11$, and we
use the notation $S_3^\bot:=1-S_3$.

\begin{Lemma}\label{lemma_N}
(a) The function
$$
\R_+\ni\lambda\mapsto N(\lambda):=\F_0(\lambda)vS_3^\bot\in\K(\H,\hs)
$$
is continuous, admits a limit as $\lambda\searrow0$, and vanishes as
$\lambda\to\infty$. The multiplication operator $N:C_{\rm c}(\R_+;\H)\to\Hrond$ given
by $(N\xi)(\lambda):=N(\lambda)\xi(\lambda)$ for $\xi\in C_{\rm c}(\R_+;\H)$ and
$\lambda\in\R_+$, extends continuously to an element of
$\B\big(\ltwo(\R_+;\H),\Hrond\big)$.

\medskip
\noindent
(b) The function
$$
\R_+\ni\lambda\mapsto\widetilde N(\lambda)
:=\F_0(\lambda)v\lambda^{-1/4}S_3\in\K(\H,\hs)
$$
is continuous and vanishes as $\lambda\searrow0$ and $\lambda\to\infty$. The
multiplication operator $\widetilde N:C_{\rm c}(\R_+;\H)\to\Hrond$ given by
$(\widetilde N\xi)(\lambda):=\widetilde N(\lambda)\xi(\lambda)$ for
$\xi\in C_{\rm c}(\R_+;\H)$ and $\lambda\in\R_+$, extends continuously to an element
of $\B\big(\ltwo(\R_+;\H),\Hrond\big)$.
\end{Lemma}

\begin{proof}
The continuity of the functions $\lambda\mapsto N(\lambda)$ and
$\lambda\mapsto\widetilde N(\lambda)$, the fact that $N(\lambda)$ and
$\widetilde N(\lambda)$ vanish as $\lambda\to\infty$, and the fact that $N(\lambda)$
admits a limit as $\lambda\searrow0$, follow from the inclusion $v\in\B(\H,\H_\rho)$
and Lemma \ref{lemma_F_0}(b). The fact that $\widetilde N(\lambda)$ vanishes as
$\lambda\searrow0$ follows from the expansion \eqref{eq_exp_F_0} for $\F_0(\lambda)$
and Lemma \ref{lemma_Q}(b). The remaining claims are direct consequences of the
continuity of the functions and the existence of the limits.
\end{proof}

In the next lemma, we assume for the first time in the paper the absence of
p-resonances, that is, that $T_3=0$. Also, since the wave operators $W_\pm$ are
related by the equation $W_+=W_-S^*$, we present from now on only the calculations
needed to establish the formula for $W_-$. This amounts to consider only the plus sign
in the operators $(u+vR_0(\lambda\mp i0)v)^{-1}$ appearing below.

\begin{Lemma}\label{lemma_B}
Assume that $T_3=0$.
\begin{enumerate}
\item[(a)] The functions
$$
\R_+\ni\lambda\mapsto B(\lambda)
:=S_3^\bot\big(u+vR_0(\lambda+i0)v\big)^{-1}v\F_0(\lambda)^*\in\K(\hs,\H).
$$
and
$$
\R_+\ni\lambda\mapsto \widetilde B(\lambda)
:=S_3\lambda^{1/4}\big(u+vR_0(\lambda+i0)v\big)^{-1}v\F_0(\lambda)^*\in\K(\hs,\H)
$$
are continuous and bounded.
\item[(b)] The multiplication operator $B:C_{\rm c}(\R_+;\hs)\to\ltwo(\R_+;\H)$ given
by $(B\varphi)(\lambda):=B(\lambda)\varphi(\lambda)$ for
$\varphi\in C_{\rm c}(\R_+;\hs)$ and $\lambda\in\R_+$, extends continuously to an
element of $\B(\Hrond,\ltwo(\R_+;\H))$.
\item[(c)] The multiplication operator
$\widetilde B:C_{\rm c}(\R_+;\hs)\to\ltwo(\R_+;\H)$ given by
$(\widetilde B\varphi)(\lambda):=\widetilde B(\lambda)\varphi(\lambda)$ for
$\varphi\in C_{\rm c}(\R_+;\hs)$ and $\lambda\in\R_+$, extends continuously to an
element of $\B(\Hrond,\ltwo(\R_+;\H))$.
\end{enumerate}
\end{Lemma}

\begin{proof}
Points (b) and (c) are direct consequences of point (a). So, we only give the proof of
(a).

The continuity of the functions $\lambda\mapsto B(\lambda)$ and
$\lambda\mapsto\widetilde B(\lambda)$ follows from the inclusion
$v\in\B(\H_{-\rho},\H)$, Lemma \ref{lemma_F_0}(b), and Section \ref{sec_perturbed}.
The fact that $B(\lambda)$ and $\widetilde B(\lambda)$ stay bounded as
$\lambda\to\infty$ follows from the inclusion $v\in\B(\H_{-\rho},\H)$, Lemma
\ref{lemma_F_0}(a), and the fact that
$\lim_{\lambda\to\infty}(u+vR_0(\lambda+i0)v)^{-1}=u$ in $\B(\H)$. Finally, the fact
that $B(\lambda)$ and $\widetilde B(\lambda)$ stay bounded as $\lambda\searrow0$
follows from an application of Theorem \ref{thm_asym} with $T_3=0$.
\end{proof}

We are now ready to state the main result of this section. For that purpose, we recall
that the dilation group $\{U^+_t\}_{t\in\R}$ in $\ltwo(\R_+)$, with self-adjoint
generator $A_+$, is given by $(U^+_t\varphi)(\lambda):=\e^{t/2}\varphi(\e^t\lambda)$
for $\varphi\in C_{\rm c}(\R_+)$, $\lambda\in\R_+$ and $t\in\R$.

\begin{Theorem}[Explicit formula for $W_-$]\label{thm_wave_op}
If $V$ satisfies \eqref{eq_cond_V} with $\rho>11$, and $T_3=0$, then we have the
equality
$$
\F_0(W_--1)\F_0^*
=-2\pi i\;\!\big\{N\big(\vartheta(A_+)\otimes1\big)B
+\widetilde N\big(\widetilde\vartheta(A_+)\otimes1\big)\widetilde B\big\}
$$
where
$$
\vartheta(s):=\tfrac12\big(1-\tanh(\pi s)\big)
\quad\hbox{and}\quad
\widetilde\vartheta(s):=\tfrac12\big(1-\tanh(2\pi s)-i\cosh(2\pi s)^{-1}\big),
\quad s\in\R.
$$
\end{Theorem}

\begin{proof}
It can be shown as in the proof of \cite[Thm.~2.6]{RT13} that there exists a dense
set $\D\subset\Hrond$ such that \eqref{start} holds for $\varphi,\psi\in\D$. We can
thus write
\begin{align*}
&\big\langle\F_0(W_--1)\F_0^*\varphi,\psi\big\rangle_\Hrond\\
&=-\int_\R\d\lambda\,\lim_{\varepsilon\searrow0}\int_0^\infty\d\mu\,
\big\langle\F_0(\mu)v\big(u+vR_0(\lambda+i\varepsilon)v\big)^{-1}v\F_0^*
\delta_\varepsilon(L-\lambda)\varphi,(\mu-\lambda+i\varepsilon)^{-1}
\psi(\mu)\big\rangle_\hs\\
&=-\int_\R\d\lambda\,\lim_{\varepsilon\searrow0}\int_0^\infty\d\mu\,
\Big\{\big\langle\F_0(\mu)vS_3^\bot\tfrac1{\mu-\lambda-i\varepsilon}
S_3^\bot\big(u+vR_0(\lambda+i\varepsilon)v\big)^{-1}v\F_0^*
\delta_\varepsilon(L-\lambda)\varphi,\psi(\mu)\big\rangle_\hs\\
&\quad+\big\langle\F_0(\mu)v\mu^{-1/4}S_3\tfrac{\mu^{1/4}\lambda^{-1/4}}
{\mu-\lambda-i\varepsilon}S_3\lambda^{1/4}
\big(u+vR_0(\lambda+i\varepsilon)v\big)^{-1}v\F_0^*
\delta_\varepsilon(L-\lambda)\varphi,\psi(\mu)\big\rangle_\hs\Big\}.
\end{align*}
Then, we can prove as in \cite[Thm.~2.5]{RT13_2} that the first term reduces to
$$
\big\langle-2\pi iN\big(\vartheta(A_+)\otimes1\big)B\varphi,\psi\big\rangle_\Hrond,
$$
and we can prove as in \cite[Thm.~2.6]{RT13} that the second term reduces to
$$
\big\langle-2\pi i\widetilde N\big(\widetilde \vartheta(A_+)\otimes1\big)
\widetilde B\varphi,\psi\big\rangle_\Hrond.
$$
So, we get the equality
$$
\big\langle\F_0(W_--1)\F_0^*\varphi,\psi\big\rangle_\Hrond
=\big\langle-2\pi i\;\!\big\{N\big(\vartheta(A_+)\otimes1\big) B
+\widetilde N\big(\widetilde \vartheta(A_+)\otimes1\big)\widetilde B\big\}\varphi,\psi\big\rangle_\Hrond,
$$
which implies the claim due to the density of $\D$ in $\H$.
\end{proof}

The formula of Theorem \ref{thm_wave_op} can be recast into a different form by
performing some commutations:

\begin{Corollary}\label{cor_vanish}
If $V$ satisfies \eqref{eq_cond_V} with $\rho>11$, and $T_3=0$, then we have the
equality
\begin{equation}\label{eq_recast}
\F_0(W_--1)\F_0^*
=-2\pi i\;\!\big\{\big(\vartheta(A_+)\otimes1_\hs\big)N B
+\big(\widetilde\vartheta(A_+)\otimes1_\hs\big)\widetilde N\widetilde B\big\}+K
\end{equation}
with $K\in\K(\Hrond)$. In addition, the functions
\begin{equation}\label{eq_2_functions}
\R_+\ni\lambda\mapsto N(\lambda)B(\lambda)\in\K(\hs)
\quad\hbox{and}\quad
\R_+\ni\lambda\mapsto\widetilde N(\lambda)\widetilde B(\lambda)\in\K(\hs)
\end{equation}
are continuous and vanish as $\lambda\searrow0$ and $\lambda\to\infty$.
\end{Corollary}

\begin{proof}
By taking into account that the functions $\lambda\mapsto N(\lambda)$ and
$\lambda\mapsto\widetilde N(\lambda)$ have limits at $0$ and $\infty$, and that the
functions $s\mapsto\vartheta(s)$ and $s\mapsto\widetilde\vartheta(s)$ have a limit at
$-\infty$ and $\infty$, one can show the inclusions
\begin{align*}
&N\big(\vartheta(A_+)\otimes1\big)
-\big(\vartheta(A_+)\otimes1_\hs\big)N\in\K\big(\ltwo(\R_+;\H),\Hrond\big),\\
&\widetilde N\big(\widetilde\vartheta(A_+)\otimes1\big)
-\big(\widetilde\vartheta(A_+)\otimes1_\hs\big)\widetilde N
\in\K\big(\ltwo(\R_+;\H),\Hrond\big).
\end{align*}
This type of result has already been proved in \cite[Lemma~2.7]{RT13} and is based on
an argument of Cordes, see for instance \cite[Thm.~4.1.10]{ABG}. These commutation
relations and Theorem \ref{thm_wave_op} imply \eqref{eq_recast}. The continuity of the
functions \eqref{eq_2_functions}, as well as the equalities
$$
\lim_{\lambda\searrow0}\widetilde N(\lambda)\widetilde B(\lambda)
=\lim_{\lambda\to\infty}\widetilde N(\lambda)\widetilde B(\lambda)
=\lim_{\lambda\to\infty}N(\lambda)B(\lambda)
=0,
$$
follow from Lemmas \ref{lemma_N} and \ref{lemma_B}(a). Finally, the equality
$\lim_{\lambda\searrow0}N(\lambda)B(\lambda)=0$ follows from Theorem
\ref{thm_S_matrix}, the equality
$\lim_{\lambda\searrow0}\widetilde N(\lambda)\widetilde B(\lambda)=0$ and the identity
(see \eqref{eq_S_matrix})
$$
N(\lambda) B(\lambda)+\widetilde N(\lambda)\widetilde B(\lambda)
=-\tfrac1{2\pi i}\big(S(\lambda)-1_\hs\big).
$$
\end{proof}

Let us mention that if $S_3=0$, then the formula \eqref{eq_recast} can be simplified
further. Indeed, in such a case one has $\widetilde N\widetilde B=0$ and
$-2\pi iN(\lambda)B(\lambda)=S(\lambda)-1_\hs$. Thus,
$$
\F_0(W_--1)\F_0^*
=\tfrac12\big\{\big(1-\tanh(\pi A_+)\big)\otimes1_\hs\big\}\big(S(L)-1_\Hrond\big)+K
$$
with $L$ the multiplication operator introduced in Section \ref{sec_free_ham}. On
another hand, the fact that $\F_0$ diagonalises $H_0$ implies $S=\F_0^*S(L)\F_0$, and
a direct calculation shows that
$$
\F_0^*\big\{\big(1-\tanh(\pi A_+)\big)\otimes1_\hs\big\}\F_0=1+\tanh(\pi A/2)
$$
with $A$ the generator of dilations in $\H$. Therefore, we get that
$$
W_--1=\tfrac12\big(1+\tanh(\pi A/2)\big)\big(S-1\big)+K'
$$
with $K'\in \K(\H)$, which is exactly the formula obtained in \cite{RT13_2} in the
generic case.

In the more general case, with no assumption on $S_3$, similar computations do not
lead to such a neat formula. Indeed, the multiplication operators $-2\pi iNB$ and
$-2\pi i\widetilde N\widetilde B$ are not individually related to the scattering
operator; only their sum is related to $S(L)$ through the equality
$$
-2\pi i(NB+\widetilde N\widetilde B)=S(L)-1_\Hrond.
$$
The best we can get in this situation is therefore the following:

\begin{Corollary}\label{cor_after_com}
If $V$ satisfies \eqref{eq_cond_V} with $\rho>11$, and $T_3=0$, then we have the
equality
\begin{align*}
W_--1
&=\tfrac12\big(1+\tanh(\pi A/2)\big)\F_0^*(-2\pi iNB)\F_0\\
&\quad+\tfrac12\big(1+\tanh(\pi A)-i\cosh(\pi A)^{-1}\big)
\F_0^*(-2\pi i\widetilde N\widetilde B)\F_0+K'
\end{align*}
with $K'\in\K(\H)$.
\end{Corollary}



\end{document}